\documentclass[format=acmsmall, review=false,nonacm]{acmart}

\makeatletter
\def\@ACM@checkaffil{% Only warnings
    \if@ACM@instpresent\else
    \ClassWarningNoLine{\@classname}{No institution present for an affiliation}%
    \fi
    \if@ACM@citypresent\else
    \ClassWarningNoLine{\@classname}{No city present for an affiliation}%
    \fi
    \if@ACM@countrypresent\else
        \ClassWarningNoLine{\@classname}{No country present for an affiliation}%
    \fi
}
\makeatother
\usepackage[ruled]{algorithm2e} % For algorithms
\usepackage{enumitem}
\usepackage{amsthm}

\newcommand{\dave}[1]{\textcolor{purple}{}}

\usepackage{xspace}
\newcommand{\direct}[0] {Directly Harmed Students\xspace}
\newcommand{\indirect}[0] {Indirectly Harmed Students\xspace}
\newcommand{\directHelp}[0] {Directly Helped Students\xspace}
\newcommand{\indirectHelp}[0] {Indirectly Helped Students\xspace}
\newcommand{\Envy}[0] {Envious Students\xspace}

\SetAlFnt{\small}
\SetAlCapFnt{\small}
\SetAlCapNameFnt{\small}
\SetAlCapHSkip{0pt}
\IncMargin{-\parindent}

\setcitestyle{authoryear}

\title[Post-Match Error Mitigation for Deferred Acceptance]{Post-Match Error Mitigation for Deferred Acceptance}

\author{Abraham Gale}
\affiliation{
  \institution{Rutgers, the State University of New Jersey}
  \city{Piscataway}
  \state{New Jersey}
  \country{USA}
}
\email{abraham.gale@rutgers.edu}
\author{Amélie Marian}
\affiliation{
  \institution{Rutgers, the State University of New Jersey}
  \city{Piscataway}
  \state{New Jersey}
  \country{USA}
}
\email{amelie.marian@rutgers.edu}
\author{David M. Pennock}
\affiliation{
  \institution{Rutgers, the State University of New Jersey}
  \city{Piscataway}
  \state{New Jersey}
  \country{USA}
}
\email{david.pennock@rutgers.edu}

\begin{abstract}
Real-life applications of deferred-acceptance (DA) matching algorithms sometimes exhibit errors or changes to the matching inputs that are discovered only after the algorithm has been run and the results are announced to participants. Mitigating the effects of these errors is a different problem than the original match since the decision makers are often constrained by the offers they already sent out. We propose models for this new problem, along with mitigation strategies to go with these models. We explore three different error scenarios: resource reduction, additive errors, and subtractive errors. For each error type, we compute the expected number of students directly harmed, or helped, by the error, the number indirectly harmed or helped, and the number of students with justified envy due to the errors. Error mitigation strategies need to be selected based on the goals of the administrator, which include restoring stability, avoiding direct harm to any participant, and focusing the extra burden on the schools that made the error. We provide empirical simulations of the errors and the mitigation strategies.
\end{abstract}

\begin{document}

% Title page for title and abstract only.
\begin{titlepage}
\maketitle

% Optionally include a table of contents
% \vspace{-0.54cm}
% \setcounter{tocdepth}{2} % adjust to 1 if desired
% \tableofcontents

\end{titlepage}

% Paper body
\section{Introduction}
Matching algorithms, like the Deferred Acceptance (DA) mechanism used to assign high-school students to New York City public schools, are designed as batch procedures that take student and school preferences as inputs and produce a single bipartite graph matching students to schools as the output. Even when preference elicitation is iterative, the output of the assignment is at some point finalized and students are sent offers of enrollment. In practice, the process is rarely so clean. Students move in or out. Implemented code may be incorrect.
Schools may shut down or lose capacity after the match. School administrators may make mistakes, for example losing applications, inputting data incorrectly, or inadvertently ranking students too highly by putting undue weight on some criteria. Errors have indeed occurred in large real-world matches, including in the New York City public school system~\cite{WJSError}. Simply re-running the matching algorithm is neither practical nor fair once an assignment has been announced and students and schools have begun planning. Institutions need criteria and procedures for recovering from errors. Restoring trust in the decision processes is one important part of accountability; another is to mitigate the effects of any error that may arise. Any system is bound to have errors that can be due to a variety of reasons: bad modeling, implementation mistakes, or data entry errors. Policymakers should account for the eventuality that their decision systems will sometimes fail and plan for how to recover from the error as best as possible.

A recovery process will start by recalculating the correct decision with corrected inputs. For example, in the school admission error scenario, the first step is to compute the correct match with the errors fixed. The next step is to identify which parties (students) were harmed by the error. We propose solutions to quantify the effects of an error and mitigate its repercussions on affected parties.  

In this paper, we make the following contributions. We
\begin{enumerate}
    \item categorize several different types of errors in school matching, focusing on three: resource reduction, subtractive errors, and additive errors;
    \item categorize the groups of students affected by errors according to direct harm, indirect harm, and envy, we also identify students who were directly or indirectly, unexpectedly helped by the error;
    \item add error-mitigation goals to the list of traditional matching goals;
    \item compute the expected number of students affected in each way by the error using rejection chains;
    \item define practical restrictions on recovering from errors, informing several different choices for mitigation strategies for each type of error; and
    \item analyze simulations to show how these errors and mitigation strategies might work in practice.
\end{enumerate}

\section{Background}

\subsection{Preliminaries}\label{sec:prelim}
The problem of school matching has interested researchers for many decades. One of the first well-documented applications was matching residents to hospitals starting in the 1950s~\cite {roth2003origins,mullin1951internship}.

The literature generally formulates the school matching problem~\cite{abdulkadirouglu2009strategy,abdulkadiroǧlu2020efficiency} as some variation of the following description.
%$$

Given 2 groups of actors:
\begin{enumerate}
    \item A finite set of students T, and
    \item A finite set of schools S;
\end{enumerate}
with the following attributes:
\begin{enumerate}
    \item Each student \textit{t} has a (possibly incomplete) strict preference ordering of the schools,
    \item Each school \textit{s} has a strict preference ordering of all students,
    \item Each school has a capacity $t_c$ which is the maximum number of students they can enroll, and
    \item A school $U$ which has infinite capacity and is appended to the end of every preference list indicating that a student is unmatched;
\end{enumerate}
create a bipartite matching of students to schools such that each student is matched to exactly one school, or remains unmatched.\footnote{There are several different ways to incorporate unmatched students into this framework, including using the ``unmatched'' school $U$, or considering them matched to themselves} This matching is chosen to be optimal according to some measures, as discussed below. We will use the notation $t_1 <_s t_2$ to mean that school $s$ prefers student $t_2$ to student $t_2$. The identical notation will be used for school preferences.  

%$$

In most real-world solutions to this problem, especially relating to school matching, the algorithm chosen is Deferred Acceptance (DA). DA was famously described by \citet{gale1962college}. 
DA is chosen for two important properties. First, student-proposing deferred acceptance, the type most commonly implemented provides student-optimal stable matches. This means that no student prefers a school to their current school if they have higher priority than any of the school's current students. Second, student-proposing DA is strategy-proof for the students, meaning that expressing true preferences weekly dominates any other strategy \cite{roth1982economics}. Student-proposing DA  is the only stable mechanism that also respects improvements~\cite{balinski1999tale}, which means that a student will always be weakly better off if their ranking improves among one or more schools. In addition, a number of works have proven that, with very general limitations, no strategy-proof mechanism Pareto improves on DA. This applies no matter what tie-breaking rule is used and whether or not there is an outside option~\cite{abdulkadirouglu2009strategy,kesten2010school,kesten2019strategy}. 

Using the previous definitions of school matching, Student Proposing DA outputs a match in which every student $t_n$ is matched to their highest ranked partner with the constraint that in the entire match, there are no \emph{blocking pairs} in the entire match. Note that this may mean leaving some students unmatched

A \emph{blocking pair} is a student and a school who both prefer to abandon their current partner for each other.
When $t_1$ is assigned to $s_1$ and $t_2$ is assigned to $s_2$, $t_2$ and $s_1$ form a \emph{blocking pair} if they satisfy the following two inequalities: $t_2 >_{S1} t_1$ and $s_1 >_{t2} s_2$. Note that $s_2$ may be the special school with infinite capacity appended to the end of every student's list indicating an unmatched student, and equally $t_1$ may simply be an empty seat in $s_1$.

DA, or a version of it, is used for a wide variety of school matching applications including the National Resident Matching Program (NRMP)~\cite{roth2003origins}, the Israeli
“Mechinot” gap-year programs~\cite{gonczarowski2019matching}, and French university admissions~\cite{parcoursup}. It is also used in the public school system of several major cities including Boston and New York~\cite{abdulkadirouglu2005new,abdulkadiroglu2006changing}.

\subsection{Motivations}
\label{sec:motivations}

Deferred Acceptance and other matching algorithms have been tremendously successful at simplifying the complex process of assigning students to school in large systems. However, as is often the case in real-life algorithmic deployments, sometimes the process does not go as planned, due to implementation errors, mistakes when entering inputs, or changes in the system post-match. 

In Spring 2019, the NYC school match process was marred when errors in the ranking
lists of several middle and high schools were discovered~\cite{WJSError} after students had received their assignments. Because NYC  school
admissions use DA, any error in the match has ripple effects that can potentially impact the school assignments of a proportion of the approximately 75,000 students who are matched to high schools (or middle schools) in any given year. University admissions in France in 2019 faced a similar conundrum when an unclear user interface led 400 programs to over-admit students, sometimes by a factor of ten~\cite{ParcoursupError}. French officials had no choice but to rerun the match and retract admission decisions.

Errors are inevitable in large systems. Algorithm and system designers should account for the possibility that their systems may fail and plan for how to recover from the error as best as possible. In this paper, we focus on techniques to recover from various types of errors that may occur in large matching systems. We focus on cases where re-running the matching process from scratch is not possible as decisions have already been made public and retracting them would create undue complications. 
 For instance, in a public housing allocation scenario, families may have moved into their apartments; it would not be fair or ethical to evict them because they had wrongly been assigned the unit. Similarly, students who have received admission offers to schools have already made plans, and should not have their offers retracted. The questions we face are then how to mitigate the impacts of the errors for the parties who were harmed, under the new constraints that some resources have already been assigned, and where to allocate new resources, if needed.

\subsection{Related Work}\label{sec:related}
\label{sec:rejection}

The problem of recovering from the effect of an error in the original match or from a change in resources post-match is closely related to past work on rejection chains in matching markets. Rejection chains study the impact of one rejection on subsequent matches: if a student is rejected from their top choice, they will be matched to one of their lower-ranked choices, taking an available spot from a student who would have been matched to that lower choice had the first student gotten their original match. Each rejection, therefore, creates a chain of rejections. In our setting, when rejections are caused by an error, the rejection chain represents how far the error propagates through the system: for example, when one student is mistakenly accepted to a school, another student will then be mistakenly rejected and will displace another student in a rejection chain. Conversely, the student mistakenly accepted will create an opening in the school to which they should have matched, potentially creating an acceptance chain.

The literature has explored the relationship between the length of a rejection chain and the size of the system.
One of the earliest explorations of this idea was done by \citet{kojima2009incentives} in the context of incentive compatibility. It is well known that student-proposing DA is not incentive-compatible from the school's point of view. They address under what conditions schools are able to manipulate their preferences to get a better match. In that setting, rejection chains are important since a school can only change its outcome by rejecting a student if the rejection chain is long enough to return to the school. This work makes some assumptions that may not apply in our case. For example, as acknowledged in the paper, the conclusions may not apply in a more realistic case where students have correlated preferences by region. However, the general insights of this work are relevant to our setting. Specifically, the authors argue that since there are likely to be many colleges with seats acceptable to students, rejection chains are likely to be relatively short even in large markets.

\citet{ashlagi2014stability} address the problem of stable matching in the hospital-resident scenario with couples searching for complementary matches. Their work uses rejection chains to argue about the length probability of couples' preferences interfering with the possibility of a stable match in a hospital. Along their way to their argument about stability, they bound the number of schools influenced by a rejection chain as growing as $O(\log n \frac{\lambda}{(\lambda - 1)}$) where there are $\lambda n$ positions and $n$ students. Their results are not directly relevant to our setting, since their assumptions, such as uniformly distributed extra positions and one position per hospital, are violated in the school matching markets that we are considering. However, their results provide an intuition as to how long rejection chains are likely to be. 

\citet{ashlagi2017unbalanced} use rejection chains to analyze another problem, this time the advantage given by an unbalanced market (more seats than students or more students than seats). They show, using the length of rejection chains, that an even slightly unbalanced market gives a large advantage to one side. However, as in the previous work, these results do not apply to our setting as a key assumption is that the students' preference lists are random and uncorrelated.  In a non-random market, where there is correlation between students' lists,  rejection chains can easily be at least as long as the number of schools. For example, consider a fully correlated market where all participants agree on the rank of all schools and students. In such a market, any rejection from the top school will lead to a rejection chain the entire length of the market, since each student will be knocked down to the school below until they reach the bottom schools, which are the only ones with seats available. %Because of this, in real-world situations, the solutions proposed for errors such as the ones we discuss tend to not consider the students in rejection chains when trying to fix errors. 

\citet{blum1997vacancy} show that regular DA can re-stabilize a market after the addition or subtraction of agents in the one-to-one case, intended to model the retirement of existing workers. \citet{boyle2009sequential} show that when agents enter the market sequentially, the later entrants are better positioned. 

\citet{feigenbaum2020dynamic} deal with a problem inverse to the one in this paper: When some students choose to reject their assigned match and leave the system, there are additional seats available to students post-match. They suggest these wait lists use a reverse of the initial lottery to assign these new seats.

\citet{mai2018finding} deal with the related problem of pre-empting errors in input by choosing a stable match that is likely to remain stable even after an error. They deal with the stable marriage case, where each school is only matched to one student. They give a polynomial time algorithm to maximize the probability of robustness for upward shifts, where an actor is moved up another actor's preference list. Their notion of upward and downward shifts is closely related to our notion of additive and subtractive errors. 

Several recent results show that introducing the possibility of expanding the capacity of schools in DA makes the matching problem both computationally difficult and in some cases impossible to solve fairly. \citet{yahiro2020game} introduce a problem where the school capacities are allocated based on need. Instead of each school having a fixed capacity, the system has resources that can be assigned to each school. Within this framework, they show that no mechanism satisfies basic efficiency requirements while being fair and strategy-proof. \citet{bobbio2022capacity} show that the problem of how to optimally expand the capacity of a school to increase the student match is NP-Complete.

\section{Problem Statement}
\label{sec:problem}

\subsection{Traditional Matching Goals}

School matching algorithms typically focus on several goals: 

\begin{enumerate}
    \item  \label{}\textbf{Student Efficiency}: The algorithm should maximize the welfare of the students, as defined by matching them to a school as high up their preference list as possible.
    \item  \textbf{Fairness}: Students should not be promoted above each other because of factors that are not central to the working of the algorithm.
    \item  \textbf{Stability}: There should be no student-school pair that both prefer each other to any of their current partners.
    \item  \label{}\textbf{School Efficiency}: The algorithm should maximize the welfare of the schools, as defined by matching them to students as high up their preference list as possible.
\end{enumerate}

Traditional matching algorithms must choose which of these goals to prioritize, as some goals are not compatible. For example, Student Efficiency and Stability, or Student Efficiency and School Efficiency are known to be incompatible in many cases. 

When an error occurs, some of the original properties of the matching system will inevitably be lost, as the error introduces noise that impacts the guarantees of the original match. An error mitigation algorithm must consider which properties are worth trying to preserve and which ones can be relaxed.

\subsection{Error Mitigation Goals}
In the error mitigation context, several new goals may be required to address the impact of the error and the constraints of the new situation: 
\begin{enumerate}
    \item \textbf{Do No Harm}: A student who is not directly harmed by the error should not be additionally harmed by the error mitigation. \label{goal:no_harm}
    
    \item \textbf{Offer preservation}: A student given an offer should not have that offer rescinded by the error mitigation strategy (some school systems will not rescind offers by policy). \label{goal:offer}
    \item \textbf{Augmentation Minimization and Fairness}: When new resources need to be added to a system to meet the newly created need, the amount of resources added should be minimized and fairly distributed.
    \item \textbf{Harm Focus}: When the error is caused by a particular party, the harm caused by the error mitigation should be focused on that party.\label{goal:focus}
\end{enumerate}

Error mitigation algorithms must consider and prioritize these new goals to design effective error mitigation solutions.

As in the case of matching goals, no mitigation technique will be able to satisfy all the error mitigation goals in all cases, and choices will need to be made. Additionally, real-world concerns can make meeting these goals difficult. To take a real-world example, in the case when a top school in a particular city caused an error by accepting students out of priority order, new seats were opened in the school to accommodate the students directly harmed by the error. This seemed to be a reasonable solution to minimize the impact of the error without further harming students. However, many of the students who were offered the new seats refused them, since the reduced per-student resources that the school now offered made it now less desirable than their erroneously assigned choice, even though they had preferred before the original match. This meant that these students were now displacing students at the school to which they erroneously matched. This shows that even among students who were not impacted directly by the error at the original school, additional harm was done to them by the error and by the efforts to address it.

\subsection{Affected Students}

To develop a framework in which we can discuss mitigation strategies we first must consider how students were harmed by the error. We focus our analysis on school matching scenarios, where harm reduction for school does not make as much sense: it is critical that students are treated fairly and efficiently; school efficiency is typically less important.

First, we will deal with the \textbf{direct} effects of the error. We will start with Directly Harmed students. Intuitively, these students are harmed as a direct effect of the error, not by the rejection chain caused by other students. 

For the formal definition, we need a concept of who is \textbf{affected} by an error. This will be different for different types of errors and will be discussed for each type of error. We will call the group of those students $D$. Let $s_t$ be the school that student $t$ is matched to in an error-free run of DA. Let  $A_t$ be the school that student $t$ was actually matched to in the run with errors. Both $A_t$ and $s_t$ can be $U$, the school indicating an unmatched student.
\begin{definition}
\label{students:direct} Let a student $t$ be \textbf{Directly Harmed} if $ A_t <_t s_t$ and $t\in D$
\end{definition}

We can similarly define a directly helped student as
\begin{definition}
\label{students:help_direct}  Let a student $t$ be \textbf{Directly Helped} if $ s_t <_t A_t$ and $t\in D$
\end{definition}

We can then move on to the \textbf{indirect} effects of the error.  Intuitively, these students are affected by the rejection or acceptance chains caused by other students.  

We then define indirectly harmed students as follows:
\begin{definition}
\label{students:indirect} Let a student $t$ be \textbf{Indirectly Harmed} if $ A_t <_t s_t$ and $t\notin D$
\end{definition}

We can similarly define an indirectly helped student as
\begin{definition}
\label{students:help_indirect} Let a student $t$ be \textbf{Indirectly Helped} if $s_t <_t A_t$ and $t\notin D$
\end{definition}

Our last group is not a new one, but one previously defined as part of DA. 

\begin{definition}
    \label{students:jealous}
    Let a student $t$ be an \textbf{Envious Student} if they have justified envy of at least one other student. Specifically,  Either $\exists t_i$ s.t.  $t >_{A_{ti}} t_i$ and $A_{ti} >_{t} A_t$. Or alternatively $\exists s_i$ s.t. $s_i$ is not filled to capacity and $s_i >_{t} A_t$
\end{definition}

With no errors DA will not lead to any Envious Students, however, errors can lead to students with justified envy. 

A student is directly harmed by the error and thus a one of the \direct, they will often also be one of the \Envy.

These measures do not include all types of students who might be affected by the error. In particular, it ignores students that are harmed by the error in ways unrelated to their choice list. For example, a student may end up in the same school as in the non-error case but with a weaker peer group due to the error. However, we believe that this list includes the main groups of students that should be considered when fixing an error in practice, so that trust in the matching process, and belief in its fairness, is maintained.

We will compute the expected size of each of these groups in each of the error scenarios considered in this paper.

\subsection{Types of Errors}

We consider three separate types of errors, which each exhibit different behaviors in terms of the students being harmed. 
\begin{enumerate}
    \item \textbf{Resource Reduction}  involves removing a school from the system after the match has occurred. A real-world example of this kind of error is when a school closes after the match but before the school year starts. We define `affected' here as having previously been assigned to the closed school before it closed. Therefore, \direct are the students who were originally matched to the school that closed; these students will also usually be part of the \Envy. Before mitigation, no students other than the \direct fall into any of our categories of affected students. 
    \item \textbf{Subtractive Errors} involve a school (or more) ignoring candidates who should have been considered. A real-world example of this is a school losing a packet of applications and not including them in their ranked list.  We define `affected' here as being one of the students who are ignored by the school and who should have been admitted in an error-free run. In this case, \direct are the students who should have been matched to the school but whose applications were lost, \indirect are the students who should have been matched to the schools to which the \direct were matched (and down the rejection chain). There will also be \indirectHelp. For example, if the school losing the applications has a large applicant pool, some students will be admitted when they should have been rejected in favor of the \direct; these students are \indirectHelp. 
    \item \textbf{Additive Errors} involve errors in the ranked list of some schools that result in some students being ranked higher than they should have been for that school, and therefore being offered a seat ahead of some students who should have been ranked higher.
    A real-world example of this category of error is accidentally miscomputing the score of a subset of students.  We define affected students here as the students who are ranked higher. In this case \directHelp are the students who are ranked higher and attend that school as a result, \indirect are the students who should have been matched to the schools to which the \directHelp were matched (and all other students harmed by the continuing rejection chain). As we will prove below, this leads to an extremely large number of \Envy. 
\end{enumerate}

We now analyze the effect of each error and provide specific error mitigation strategies.

\section{Resource Reduction}
\begin{center}
\renewcommand{\arraystretch}{2}
\begin{table}[]
\begin{tabular}{ |c|c|c|c|}
\hline
&\textbf{Directly Harmed}&\textbf{Directly Helped}&\textbf{Envious}\\
\hline
 \textbf{Resource Reduction}& $C$& 0& $\leq C$ \\ 
 \hline
 \textbf{Subtractive Error}& $C \times p$ & 0&$C \times \frac{p}{1-p}$ \\
 \hline
 \textbf{Additive Error}& $0$ & N & $N + (A-C-N)\times \frac{N}{N+1}$    \\
 \hline
\end{tabular}
 \caption{\label{tab:expected} The expected size of each group in each error scenario before any mitigation is performed}
\end{table}
\end{center}

A post-match resource reduction is caused by a change in the algorithm environment that 
unfairly harms a subset of students. This can happen if, for example, a school closure is decided after the match, or if an incident makes the school building unusable. Students matched to that school lose their match and the possibility of applying to other desired schools as these have already been assigned students through the original match.

We define this error as occurring after running DA over a group of schools $S$, and students $T$. We examine the case when one school, $s_e$, among all schools $S$ closes after the match. We then consider all the students matched to $s_e$ to be unmatched and remove $s_e$ from $S$. All other students and schools remain the same as assigned in the original DA run.

Formally, 
\begin{definition}
    A Resource reduction error occurs by using DA to produce a match on (S,T), then removing one school $s_e$ from $S$. All students matched to $s_e$ are matched to $U$
\end{definition}

\subsection{Size of Each Affected Students Group}

In this case, only the directly affected students (the ones matched to $s_e$ in the original match) are affected at all by the error. As shown in Table~\ref{tab:expected}, this means there are exactly $C$ \direct, where $C$ is the capacity of the error school. 
Since these $C$ are now unmatched, no other student would prefer to take their places, and they do not interfere with the rest of the match. This makes the number of \indirect zero. Since the original match is assumed to have been run using DA, there is no justified envy among the students not directly impacted by the error. Removing this school does not help any students directly or indirectly.

Unless the \direct would have been unmatched in a match run over $S - s_e$, they will also be \Envy, since they will be ranked in at least one school $s$ above at least one other student who currently has a spot, and all students prefer a spot on their list to being unmatched. Therefore, the number of \Envy will be approximately $C$ 
if $s_e$ is a popular school, or slightly smaller if $s_e$ is not a popular school (and thus some \direct would have been unmatched had they not matched to $s_e$).

\subsection{Mitigation Model}

As discussed above, offer preservation means that there is a strong preference against moving students when they are already assigned to a school. In this case, where all the harmed students are not assigned to any school, we can minimize the impact on \direct by only moving the \direct without worrying that they would prefer their new erroneous assignment. 

We will consider the following restrictions to address resource reduction:
\begin{enumerate}
    \item No student who was not displaced should be moved during mitigation \label{goal:stable_expansion:move} (\textbf{Offer preservation})
    \item The post-mitigation match should be stable \label{goal:stable_expansion:stable}  (\textbf{Stability})
    \item All students should either be matched to a school on their preference list, or remain unmatched. 
    \label{goal:stable_expansion:prefs}  (\textbf{Student efficiency})
\end{enumerate}
Restriction \ref{goal:stable_expansion:stable} aligns well with the preferences of the system designers given their pre-error choices. By choosing deferred acceptance, they showed their inclination toward stability over maximizing Pareto-efficiency.  Restriction \ref{goal:stable_expansion:prefs} is likewise simply a restriction of the original match. Restriction \ref{goal:stable_expansion:move} is a natural outgrowth of the overall error correction goal to minimize the disruption to the system as a whole, and not harm any further students. Additionally, in real-world systems, there is a concern that a currently matched student may change their preferences after the error and prefer to not move, which can cause even seemingly beneficial changes to upset some students.

Note that none of these
restrictions prevent the addition of extra seats at any school. Adding resources in this case may be necessary due to the loss of resources caused by the original error.

 This problem is almost identical to the goal of the original DA algorithm: find the student-optimal stable match. The difference is that we no longer consider the capacities of the schools as fixed so that the system can absorb students impacted by the error, and that we are bound by all offers already given to other students.

Formally, we define the problem for resource reduction is to produce a match that, like the original DA match, is a student-optimal stable match. 
 \begin{definition}
 \label{prob:resource_reduction}
      A solution to the resource reduction problem should produce a matching from S to T such that each student $t_n$ is matched to their best stable partner maintaining all matches not dissolved by the closing of the school and allowing any school to expand its capacity.  
 \end{definition}
 This definition is saying that the goal is to replicate DA as far as possible while maintaining all previous offers and sacrificing capacity limitations.
\subsection{Stable Expansion}
\begin{algorithm}
\caption{Stable Expansion Algorithm}\label{alg:stable_expansion}
$D \gets \direct$\;
\For{$t \in D$}{
    \For{$0 < i < Preference List_t.length$}{
    $s \gets Preference List_t[i]$ \;
    \If{$t >_{s} Br_{s}$}{
      Match(t,s)\;
      Break\;
      
    }

  }
}

\end{algorithm}
For this error case, the problem we laid out in Section \ref{prob:resource_reduction} can be solved. We call the algorithm that finds this student-optimal stable match {\em Stable Expansion}. {\em Stable Expansion} is a modification of DA that maximizes student welfare subject to the restrictions above. To achieve this, {\em Stable Expansion} ensures that no \direct are left with justified envy, by guaranteeing that all \direct get a seat in any school on their list for which they rank higher than any student rejected by the school.

Let the best rejected student in $S$ be $Br_{S}$, or the highest-ranked student not in \direct that applies to school $S$ during deferred acceptance and is rejected. If no student outside of the \direct is rejected from $S$, this is a hypothetical student who is ranked lower than all real students by every school. Intuitively, this is the student with the most cause for justified envy if someone else is admitted. Let $a >_{s1} b$ iff $s_1$ ranks student $a$ above student $b$. {\em Stable Expansion} is then shown in Algorithm \ref{alg:stable_expansion}. Note that {\em Stable Expansion} does not take into account the capacity of any school; schools that are full are treated the same as schools that have space.

If the original match is done using DA, {\em Stable Expansion} meets all our restrictions, and finds the student-optimal stable match. Stable Expansion clearly meets Restriction \ref{goal:stable_expansion:move}, since no student who was not affected is even considered by the algorithm when deciding who to move. It also clearly only matches students to schools on their preference list meeting Restriction~\ref{goal:stable_expansion:prefs}, since only schools on their preference list are chosen. As we will now prove, given our constraints, {\em Stable Expansion} meets Restriction~\ref{goal:stable_expansion:stable} by producing the student-optimal stable match.
\begin{lemma}
    Stable Expansion produces a stable match.
\end{lemma}
\begin{proof}
Instability requires a \emph{blocking pair} of a student and school who both prefer to abandon their current partner for each other. 
When $t_1$ is assigned to $s_1$ and $t_2$ is assigned to $s_2$, $t_2$ and $s_1$ form a \emph{blocking pair} if they satisfy the following two inequalities: $t_2 >_{s1} t_1$ and $s_1 >_{t2} s_2$. Note that $s_2$ may be the special school with infinite capacity indicating an unmatched student, and equally $t_1$ may simply be an empty seat in $s_1$. 

To prove the stability of {\em Stable Expansion}, we consider which group of students $t_1$ and $t_2$ belong to. $D$ will refer to the set of students matched by DA to the closed school ($s_e$). 

There are then three cases we will consider to cover every possibility.

\begin{enumerate}
    \item  $t_1,t_2 \notin D$
    \item $t_2 \in D$
    \item $t_1 \in D, t_2 \notin D$
\end{enumerate}
When $t_1,t_2 \notin D$, there is no instability since they were matched using DA and were not moved by {\em Stable Expansion}. 

If $t_2 \in D$, there can be no blocking pairs. This is true since {\em Stable Expansion} does not consider capacity. The lack of consideration for capacity implies that if $t_2 >_{s1} t_1$ and $s_1 >_{t2} s_2$, $t_2$ will always be matched to $s_1$ or better.  

Lastly, if  $t_1 \in D, t_2 \notin D$, we must consider that the original match was done with student-proposing deferred acceptance. Since students propose in order of preference, and $t_2$ is not matched to $s_1$, $t_2$ was originally rejected from $s_1$. This means that $Br_{s1} \geq_{s1} t_2$.
$t_1$ was matched to $s_1$ by {\em Stable Expansion}, implying $t_1 >_{s1} Br_{s1}$. These two inequalities imply through the transitive property that $t_1 >_{s1} t_2$, contradicting one of the two blocking pair inequalities, $t_2 >_{s1} t_1$. Therefore $t_1$ and $t_2$ do not form a blocking pair. These cases cover all possibilities, so there are no blocking pairs, and {\em Stable Expansion} produces a stable match.
\end{proof}
% \subsubsection{Student Optimal}
%DP: Turn this subsection into a Theorem (or Observation) and Proof instead.
\begin{theorem}
Stable Expansion produces the student-optimal stable match, when constrained by the original offers for all non-displaced students. 
\end{theorem}
\begin{proof}
Consider an algorithm that produces a better outcome for students called $SE_b$. Since no student outside of the \direct can be moved, $SE_b$ must move at least one of the \direct to a school higher on their preference list than in the match produced by regular $SE$. Let this student be $t_b$ and the regular $SE$ school $s_r$. Let $s_b$ be the new school which is higher on the preference list that $t_b$ is not matched to by $SE$ but is matched to by $SE_b$. Let $Br_{sb}$'s $SE$ match be $s_c$. By the definition of {\em Stable Expansion}, $Br_{sb} >_{sb} t_b$, since $t_b$ is not matched to $s_b$ by regular {\em Stable Expansion}, and $t_b$ proposes in order of preference. $Br_{sb}$ is, by definition of $Br$, rejected from $s_b$ during deferred acceptance, this gives us $s_b >_{BrSh} s_c$. This means that when $SE_b$ matches $t_b$ to $s_b$, $Br_{sb}$ and $t_b$ and $s_b$ form a blocking pair. This means that any algorithm that produces a better outcome for students than $SE$ is not stable. Therefore {\em Stable Expansion} produces a student-optimal stable match.
\end{proof}

\section{Subtractive Error}

Subtractive errors model cases where some participants were not given full consideration in the original match, for instance when a school inadvertently loses some students' applications. Formally, we define this type of error as removing a set of students from the preference list of the error school $s_e$. For our analysis, we assume that each student is removed with probability $p$. We then run deferred DA, with $s_e$ rejecting any student not on its preference list.

\subsection{Size of each group}
For all group size estimations for this group, we assume that the school remains competitive, i.e.,  even after the applicant removal is performed, there are enough remaining applicants to $s_e$ that are still on the preference list of $s_e$ to fill the school to capacity. (If there are not enough applicants to fill the school, a trivial error mitigation strategy would just assign the \direct from lost applications to the school.)  

As shown in Table~\ref{tab:expected}, the number of \direct is simply the number of applicants who would otherwise have been admitted if their applications had not been lost. In probabilistic terms, this is simply the proportion $p$ of lost applications multiplied by $C$, the number of seats in $s_e$.

The number of \indirect for this case depends entirely on the length of the rejections chains. As discussed in Section \ref{sec:rejection}, this will change depending on the particular market conditions. We will give some empirical results in the simulations section.

In this case, there are no \directHelp. There will be at least some \indirectHelp in the form of students who are admitted to the error school instead of the \direct. There also may be more \indirectHelp, if removing these students begins an acceptance chain. In other words, the number of \direct gives a lower bound on the number of \indirectHelp since, with our competitiveness assumption, each student in \direct is replaced with a new accepted student.

\Envy include both \direct and students who were part of the error (e.g., their application was lost) but whose direct outcome was not changed by the error. However, they are now seeing students who ranked lower in $s_e$ being offered a match to $s_e$ ahead of them. Despite the fact that they themselves would not have matched to $s_e$, they now have justified envy. The expected number of \Envy can be directly computed.  Consider each student on the ranked list to be a single trial with a probability of failure equal to $p$, where failure indicates a student who would normally have been admitted but whose application was lost. Each of these students will now be \Envy. The number of these students is simply the number of trials taken to admit $C$ students in school $s_e$. This is exactly what the negative binomial distribution gives. Using the mean of the negative binomial distribution tells us that the Expected number of \Envy is:
$$
E(Subtractive\_Error.Envious) = C\frac{p}{1-p}
$$

Note that this number is strictly larger than the number of \direct since all these students are also \Envy. 
 
While this may look large at first, for reasonable numbers of lost applications, the number of \Envy is actually relatively reasonable. For example, as long as $p < 0.5$ we should expect that the number of \Envy is no more than twice the number of \direct. Additionally, the number of harmed students does not depend on the size of the applicant pool, only the size of the school. This is important for errors in extremely competitive schools. For example, some colleges have admissions rates under 5\% so even significant increases in class sizes cannot accommodate admitting a significant fraction of the applicant pool.

One interesting point about subtractive errors is that the number of \direct and \Envy are completely independent of the rest of the system. The only requirement for these numbers to hold true is that there are enough applicants for the error school to fill up. With that caveat, a system with thousands of schools and tens of thousands of students behaves the same way as a system with 10 schools and hundreds of students. Assuming this minimum level of popularity of $s_e$, there is also no need to consider the student preference lists when computing this expected size.

\subsection{Mitigation}
\subsubsection{Mitigation Model}
In this error case we will consider two  restrictions drawn from real-world limitations:
\begin{enumerate}
    \item All students should be given the option to remain in their currently assigned school, regardless of their original preferences. \label{goal:lost:move} (\textbf{Offer preservation})
    \item Only one round of negotiation is permitted; that is, each student will receive one proposal which they can either reject or accept.\label{goal:lost:one}
    \end{enumerate}
Restriction \ref{goal:lost:move} is a straightforward expansion of Goal \ref{goal:no_harm}, Do No Harm and Goal~\ref{goal:offer}, Offer preservation. We do not hold students to their previously expressed preferences so far as to move them without their consent because it is likely that their feelings may have changed after the original match. This change can be due to either psychological effects such as loss aversion, or more concrete ones such as signing a lease or the desire to avoid a school dealing with the fallback of an error (e.g., overenrollment). Restriction \ref{goal:lost:one} is based on real-world constraints and resolving the error in a reasonable amount of time. If students are already subject to upheaval due to the late-breaking error, we should not subject them to further upheaval by extending the window of uncertainty beyond the necessary.

\subsubsection{Mitigation Options}
Unlike in the previous case, there is no clear solution that reasonably helps all harmed students. Depending on which groups of students are most critical, different error-reduction techniques should be chosen. 

We will consider three different algorithms:
\begin{enumerate}
    \item Direct only
    \item Stability Restoration
    \item Best-of-Both
\end{enumerate}

This first option is to directly maximize welfare by only moving the directly affected students, as we do in {\em Stable Expansion}. This is the approach that was used for to address the errors in the NYC match~\cite{WJSError}. Specifically, \direct were given the option to attend the error school $s_e$. While this solution seems reasonable, it does come with drawbacks:  while this means that no \direct will be directly harmed by the error anymore, there are now two other groups to consider. Both \indirect and \Envy now are harmed in a way that is not addressed by this mitigation strategy. Second, this strategy causes the entire burden of error mitigation to fall on $s_e$. While this is good from the perspective of harm focus (Goal~\ref{goal:offer}), it causes a very large burden on school $s_e$. This burden can affect the desirability of the school to the very students we are trying to help by reducing the per-student resources. 

Another option is {\em Stability Restoration}. {\em Stability Restoration} attempts to help the \Envy, which includes \direct. One way to arrive at this algorithm is to consider is who is likely to complain. In situations where the admission functions are public, \Envy and \direct can easily find out that they were harmed by comparing themselves to a single friend with lower scores who was admitted. This is how the NYC error~\cite{WJSError} was actually detected. In contrast, \indirect would only find out they were harmed if they had a deeper knowledge of the system, and of the reach of the rejection chains.

In order to help \Envy, the mitigation must restore stability. In effect, this is saying that the error school will honor the implied admissions threshold created by the error, and will  admit any student who scores higher than the lowest scoring admitted student. A major advantage of this approach is that, as before, the only required change is admitting additional students to the school that made the error (\textbf{Harm Focus}). However, this creates an even larger burden on the $s_e$ than only admitting \direct since \direct $\subset$ \Envy, but not as much larger as we might expect. If we look at Table~\ref{tab:expected} we see that for reasonably sized errors, the difference between the expected expansion is minimal. For example, if 20\% of the applicants have their application lost, admitting only \direct would mean admitting 1.2 times the original number of students on average, while admitting \Envy, which includes admitting \direct, would mean admitting 1.25 times the original number of students on average. This slight difference may be worth it for the benefits that stability can bring. 

It is also important to note that, in this section, it would be more precise to talk about {\em Justified Envy freeness} than stability. If a student leaves a school other than $s_e$ in order to take a newly created spot in $s_e$, this may leave an open spot that another student is interested in. Given our restriction on rounds of negotiation, any error mitigation strategy will have this kind of instability. We, therefore, use No-Justified-Envy in the rest of this section instead of stability, since stability will be impossible for all reasonable mechanisms.

Another option is to focus on \direct and \indirect since they suffered actual harm. This can be done by simply offering each of them the best of two matches: their original match with the error in $s_e$ list, and the match they would have received if the error had not happened. This strategy, best-of-both ({\em BoB}), has the advantage of helping all students who suffered actual harm without increasing the number of \Envy relative to the original error scenario. 
\begin{theorem}
    Best-of-Both error mitigation does not increase the number of students with justified envy.
\end{theorem}
\begin{proof}
    To prove we do not increase the number of \Envy, recall the definition of instability from before: When $t_1$ is assigned to $s_1$ and $t_2$ is assigned to $s_2$. A blocking pair satisfies the following two inequalities: $t_2 >_{s1} t_1$ and $s_1 >_{t2} s_2$. 

 We now have two sets of affected students, $D^+$ and $D^-$ where $D^+$ was helped by the error and $D^-$ was hurt by the error. We will consider all 4 possible cases:

\begin{enumerate}
    \item $t_1, t_2 \notin D^-$
    \item $t_2 \in D^-$
    \item $t_1 \in D^-, t_2 \notin D^+$
    \item $t_1 \in D^-, t_2 \in D^+$
\end{enumerate}

If $t_1, t_2 \notin D^-$ neither student is moved by {\em BoB} and there can be no additional blocking pairs created.

If  $t_2 \in D^-$ there will not be additional justified envy created by {\em BoB}, since {\em BoB} only moves students up their preference list. Informally, moving someone to a better school never makes them more envious of other students. Formally: if the new school assigned to $t_2$ is $s_3$ then  $s_3 >_{t2} s_2$. Because of this $s_n >_{t2} s_3$ implies $s_n >_{t2} s_2$, which means that the move did not create any additional schools to consider for justified envy. In all the previous schools, $t_2 >_{sn} t_1$ will remain unchanged.

If $t_1 \in D^-, t_2 \notin D^+$ this means that after running {\em BoB} both $t_1$ and $t_2$ are in the position they would get from an error-free DA.  This means that there is no justified envy between them by stability of DA. 

Lastly, if $t_1 \in D^-, t_2 \in D^+$, $t_2$ was helped by the error. From a similar argument to $t_2 \in D^-$, there is still no additional justified envy. The final position of  $t_1$, $t_2$ is now the equivalent of running DA and then improving $t_2$'s match. Because of this we have essentially taken DA and moved $t_2$ up its preference list, which we have shown cannot create justified envy.

\end{proof}

We have now shown that {\em BoB} does not create additional blocking pairs, however, it does not fix the blocking pairs that were originally created by the error. 

 Since {\em BoB} moves up the \indirect, the pathological cases can be quite disruptive. Recall from Section~\ref{sec:rejection} that the size of a rejection chain is difficult to give a hard bound on in a real system. In some pathological cases, such as when each school has a capacity of 1 and they all agree on the ranking of students, every single student is moved by this algorithm. This means that we cannot promise, or even estimate a reasonably small expansion, even if only one school made an error. Secondly, unlike the previous solution, this leads to expansion in innocent schools not possibly violating the harm focus goal. Nonetheless, this option can make sense when the main goal is to avoid direct harm to any student through the error, and if other schools are willing to participate. This is especially true if the error was caused by the clearinghouse system and harm focus is unimportant.

In most cases, we expect one {\em Stability Restoration} and {\em BoB} to be the best option. Choosing between them requires balancing the systemic upheaval caused by best-of-both against its advantage of actually making all the harmed students whole.    

\section{Additive Errors}

Additive errors model cases where some participants were unduly given high priority (ranking) in the match, giving them an advantage, for instance when a school $s_e$ inaccurately computes the score of some students, moving them higher on its preference list and penalizing other students. Formally, let $A$ be the group of students who would propose to $s_e$ under normal deferred acceptance and let $C$ be the students that $s_e$ would normally admit. We define an additive error as $s_e$ selecting $N$ students from $A /\ C$ uniformly at random and ranking them at the top of its preference list. Then DA is run with the modified preference list. 

In this type of error, each student added to the top of $s_e$'s list displaces a single other student who should have been accepted. This means that we assume $N \leq C$. In the extreme case of $N \geq C$, the admission function will admit none of the correct applicants and instead fill the seats at random from the rest of the applicant pool.

\subsection{Size of each group}
The \directHelp are the students that this algorithm moves up, and no students are \direct or \indirectHelp. 
As in the subtractive error case, the number of \indirect will depend on the rejection chain length (Section~\ref{sec:rejection}). 

The number of \Envy is more interesting than it was when the errors were subtractive. First, every one of the \directHelp will displace one student who would otherwise have been admitted. The maximum size of the additional Envious students is $(A-C-N)$ representing all the possible \Envy who are not otherwise accounted for, that is to say, all students who applied to the error school and were not admitted or pushed out. We claim in Table \ref{tab:expected} that almost all these students are in fact \Envy. To give an intuition, consider the case where $N=1$ and only one student is erroneously admitted. All students ranked above them, who are not admitted, will have justified envy. We then can expect around half of the other students to be ranked below them, since they are admitted at random and other students are just as likely to be ranked above them as below. This problem becomes much worse as more students are erroneously admitted, since to avoid justified envy, a student must rank below all the error students. The exact expectation is proportional to $\frac{N}{N+1}$ as proved in the appendix by a combinatorial argument based on the following summation.
$$
 \sum_{j=0}^{A-C-N} j \times \frac{\binom{j + N - 1}{N-1}}{\binom{A-S}{N}} = (A-C-N) \times \frac{N}{N + 1}
$$
As before, this applies when $(A-C-N) >=0$.

\subsection{Adjustments needed for additive errors}

The mitigation options for additive errors are very similar to those for subtractive error, but require some adjustment.
\begin{enumerate}
    \item Direct only
    \item Stability Restoration
    \item Best-of-Both
\end{enumerate}

Direct only focuses on the students who were displaced by the \directHelp is possible. This can be done by simply admitting all the students who were denied a seat as a direct result of the error students taking their spots. For any error, assuming that the size of the error is less than the size of the school, the expansion required will be at most $N$, the size of the error. This mitigation strategy however does have some issues. First of all, this group is not as clearly defined as before. Why should the students who are pushed out by the first step of a rejection chain have an advantage? Both the directly and indirectly affected students are harmed by their slots being taken unjustly by someone else. Unlike with regard to subtractive errors, it is not clear that these students are directly targeted by the error. This makes it harder to justify giving them an advantage just because in the original match they should have been matched to a school that made an error.

The best-of-both ({\em BoB}) strategy from the Subtractive Error scenario is another reasonable choice, although it can here as well create significant upheaval throughout the entire system if it is susceptible to long rejection chains. Note that the property of not increasing the number of \Envy is independent of the exact error and applies here as well.

Stability Restoration becomes a much less attractive option than it was when the error was subtractive. The number of \Envy creates a situation more complex than that of the subtractive error scenario. We cannot apply {\em Stability Restoration} since even small mistakes would require admitting significant portions of the entire applicant pool to $s_e$, regardless of the capacity of the school. If the error involves even as few as 4 students, no matter the size of $C$, $A$, and $N$, restoring stability requires admitting 80\% of the remaining applicant pool.

The main distinction between Additive and Subtractive Errors is that there are too many \Envy to help in the Additive case. This means that for most cases the decision becomes whether the decision makers want to consider \indirect, in which case they should use {\em BoB}, or only the beginning of the rejection chain, in which case they should use direct only.

\subsection{Near-Stable Expansion}

Because stability is not a reasonable option in additive errors without retracting offers, we suggest an alternative of near-stability. We define this as there should be no blocking pairs where $t_1$ is not one of the mistakenly accepted students. This definition is already met by the error case, and we can set an additional restriction that any error mitigation must maintain it.  

One mitigation strategy using this idea is {\em Near-Stable Expansion}. {\em Near-Stable Expansion} is similar to {\em Stable Expansion} from the Resource Reduction scenario. The affected students are any group of students harmed by the error. The students who were erroneously admitted to the system are not considered at all for stability evaluation. Instead, the new best rejected student  $Br_S$ for each school $S$ is the best uninvolved rejected student, ignoring all the error students and all the affected students. 

If the affected students are \direct, then this will lead to a similar result to Direct Only. Since all the \direct would have been accepted to $s_e$ in the original match, they will all get an offer to $s_e$ or a school they ranked above it. The difference is that some of the students may get an offer to a school they ranked above $s_e$ if they rank above the best rejected for that school. This algorithm offers the advantage of spreading the expansion throughout the system, at the expense of harm focus.  A similar algorithm can be made to spread the expansion of {\em BoB} using both \direct and \indirect as the affected students.

\section{Simulations}
\subsection{Experimental Settings}

We performed some simulations to show empirically the sizes of each group of students and how the proposed algorithms look in practice. The general design of the preference lists for the simulations is taken from \cite{hafalir2013effective} and \cite{kesten2019strategy}. The basic idea is that each student has a popularity score drawn from a normal distribution which all schools agree on. This score is added to an individual score particular to each (student, school) pair to get the student's final score. The schools then sort the students by their final score to get a ranking. The students sort the schools in the exact same way. This combination allows for correlated but distinct rankings for each agent. We add a list length drawn from a normal distribution, then schools are drawn uniformly until the list is full. We use equal weighting for the individual and universal scores.

In each simulation, there are 10 schools with 100 seats each and 900 students. The individual score is equally weighted with the overall score. For accuracy, each simulation is run 100 times and the average across all of them is taken.

\subsection{Resource Reduction}
For this case, we ran three experiments, each with a school with different popularity in the students' lists removed:  the top school in the system, the median-ranked school, and the lower-ranked school. As shown in Table \ref{tab:resourec_reduction}, the group sizes work out as expected, with the entire size of the school as \direct. The lower-ranked schools do not fill up in some of the simulations. This means that there will be fewer than expected \direct, and \Envy. The number of \Envy depends on whether the \direct would find a match in a system without the missing school, which as expected is very likely if they are the type of students who fill the top school and very unlikely if they are the type of students who fill the bottom school. As we expect, \Envy is almost the entire group of the best students, and much less, about $\frac{3}{4}$, of the lower-ranked ones.
\begin{table}[]

\begin{tabular}{|l|l|l|l|}
\hline
School Popularity       & \direct & \indirect & \Envy \\ \hline
Popular   & 100     & 0       & 97.99   \\ \hline
Median & 99.86   & 0       & 93.39   \\ \hline
Unpopular  & 41.25   & 0       & 27.11   \\ \hline
\end{tabular}
\caption{The average size of each group when one school is removed in a resource reduction error. The school removed is either the most desired (Popular) school in the system, a median school or the least desired (Unpopular)}
\label{tab:resourec_reduction}
\end{table}

{\em Stable Expansion} is effective at reducing the number of unmatched students after an error. In Figure~\ref{fig:stable_expansion}, we see that {\em Stable Expansion} successfully places most of the students in their second-choice school. Since most of these students' first-choice school has left the system, this is the best possible outcome for them. The expansion of resources is also quite reasonable. The average total number of additional seats was 86.2. The average maximum expansion in any one school across all the runs was 26.01. Even in the worst case across any of the 100 runs, the maximum single-school expansion was only 42 seats. This shows that Stable Expansion successfully spreads the burden of the extra seats across the different schools in the system in practice while maintaining good outcomes for the displaced students. 

Figure \ref{fig:stable_expansion} shows the placement of each student on their preference list before the error, after the error, and after Stable Expansion is run. The red bars show the students who are matched to the error school so that their path through error and mitigation is clear. Since the best school is removed, we see that many of the students in \direct had their first choice after the error, which means that no matter the mitigation many of them will end up worse off. In this case, Stable Expansion usually does as well as possible placing these students at their second choice. We see that some students are even better off than before the error, this happens when they themselves were one of the best rejected students at their favorite school, and are therefore able to attend when the school is expanded since no non-affected student outranks them.

The moved section is used to indicate the number of students who changed from their DA placement under each algorithm. The colors there show which direction they moved.
\begin{figure}
    \centering
    \includegraphics[width=\textwidth]{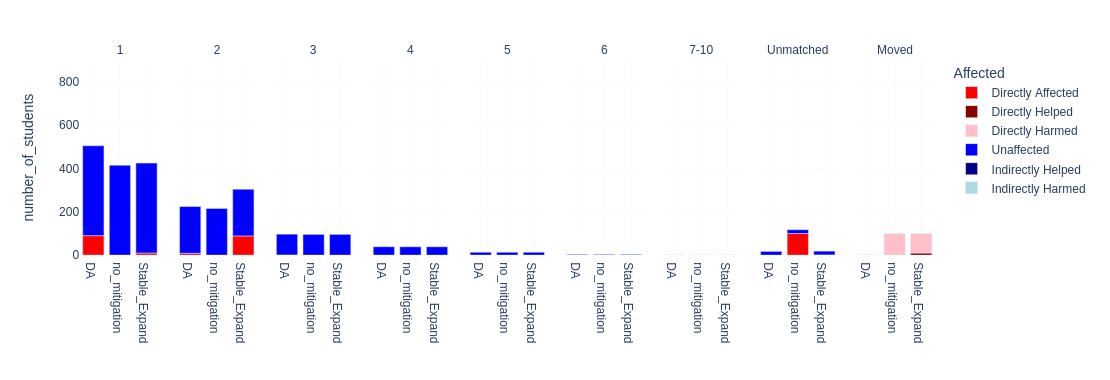}
    \caption{Affected and unaffected students change in outcome when the best school is removed, under the original condition (no removal), removal, and Stable expansion}
    \label{fig:stable_expansion}
\end{figure}
\subsection{Subtractive Errors}
In this experiment, we focused on a scenario where the top school lost applications. This will tend to be the most impactful type of error since highly rated students are more likely to start rejection chains. In our simulations, the applicant pool for the best school hovered around 320 seats. 

As we can see in Figure \ref{fig:subtractive_groups}, the predictions of the group sizes from Table~\ref{tab:expected} match the outcomes almost exactly. In this simulation, the number of \indirect is small relative to the number of \direct. This means that for this particular combination of preferences and capacities, best-of-both may be a reasonable option. As expected, we see that the number of \directHelp is lower bounded by the number of \direct. In this case the the number of \indirectHelp above the lower bound is about half as much as the number of \indirect. This indicates that the positive rejection chains are about half as long as the negative ones, We see that when the proportion of lost applications is small, roughly less than $\frac{1}{4}$, accommodating \Envy is reasonable. As the proportion of lost applications grows, this becomes less and less reasonable. Note that these results are independent of the rest of the system and only depend on the size of the error. For example, doubling the number of schools and students has no effect on the size of the groups. The only systemic change that has an effect is reducing the popularity of the error school to the point where they do not have enough applicants to fill, taking into account that some of the applications are lost. 

\begin{figure}
    \centering
    \includegraphics[width=\textwidth]{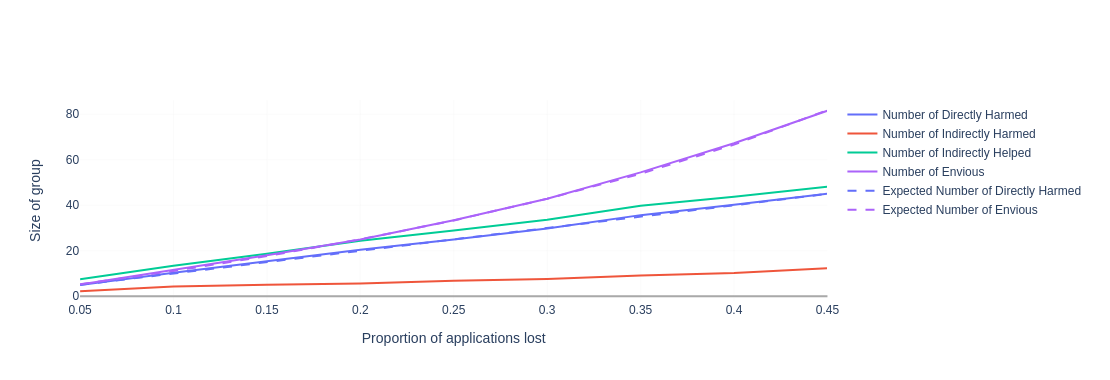}
    \caption{The actual vs expected group size as we vary the proportion of lost applications}
    \label{fig:subtractive_groups}
\end{figure}

In Figure \ref{fig:subtractive_mitigation} the effects of different mitigation strategies for a Substractive Error are shown. We see that all the strategies help the \direct shown in light red. There are a few \indirect not helped by stability restoration or direct only, shown in the last column in light blue. There are a significant number of \Envy helped only by stability restoration. Since, as we saw above, the number of \indirect is small, BoB mostly dominates direct only in this case. 

The moved section again shows the number of students who changed from their DA placement under each algorithm. For these algorithms, capacity is only added in the error school, even for the BoB algorithm for which this is not strictly true, the average max capacity is 48 seats which is almost identical to the number of students moved. For all the other algorithms, all moved students create a new seat.

\begin{figure}
    \centering
    \includegraphics[width=\textwidth]{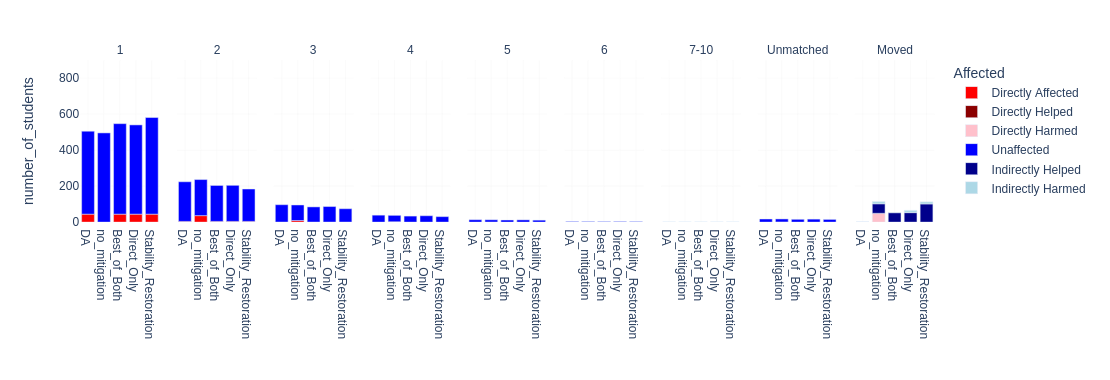}
    \caption{The effect of mitigation strategies when a subtractive error with $p=0.5$}
    \label{fig:subtractive_mitigation}
\end{figure}

\subsection{Additive Errors}

In this experiment, we simulate an additive error in the most popular school. In Figure \ref{fig:additive_groups}, we see that our estimate of Table~\ref{tab:expected} for the number of \Envy is quite accurate. Unlike with subtractive errors, the line here is not smooth since the number of \Envy itself changes with the number of applicants to $s_e$, which are randomly generated. We also see that the number of \indirect grows approximately linearly with the number of \directHelp, this means that each student in \directHelp starts a rejection chain of approximately the same length. Even our maximum error case of 20 students is only the equivalent of $p=0.2$ in the subtractive case. We limited Figure \ref{fig:additive_groups} to this error size since the number of \Envy grows quickly with the size of the error.

\begin{figure}
    \centering
    \includegraphics[width=0.9\textwidth]{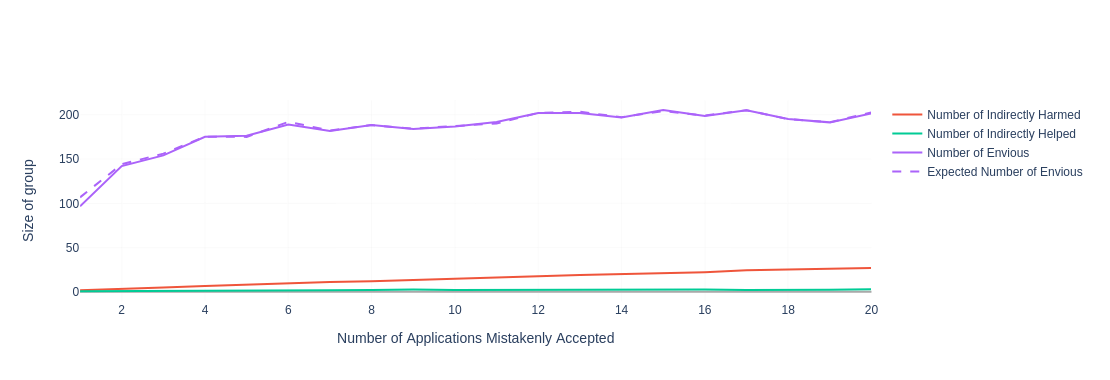}
    \caption{The actual vs expected group size as we vary the proportion of erroneously added students}
    \label{fig:additive_groups}
\end{figure}

In Figure \ref{fig:additive_mitigation}, we see the effect of a relatively small error of 10 students. As expected, {\em Stability Restoration} has a large number of students moved and extra capacity compared to the other options. However, since only 10 students were directly displaced and around 5 were indirectly displaced, options like Best-of-Both may be a good fit. As before the capacity expansion is approximately the same as the number of students moved. 
\begin{figure}
    \centering
    \includegraphics[width=\textwidth]{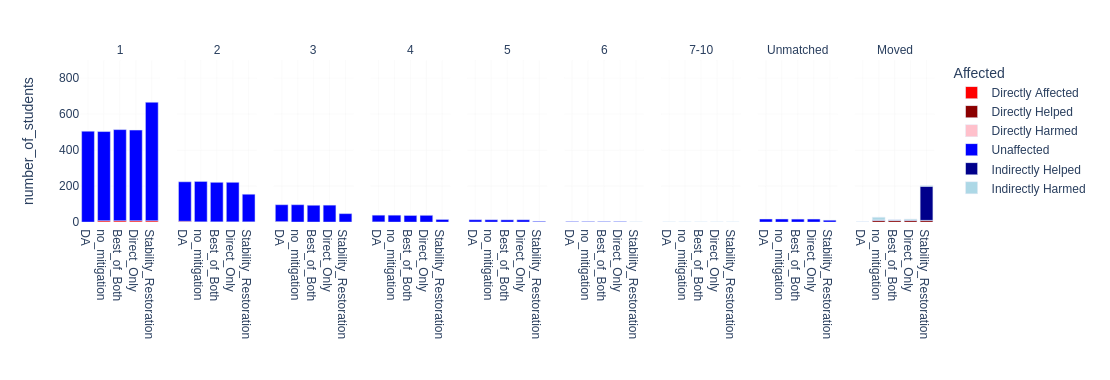}
    \caption{The effect of mitigation strategies with an Additive Error $N=10$}
    \label{fig:additive_mitigation}
\end{figure}

\section{Conclusion}
This paper shows a first step towards understanding and addressing the effects of errors, or changes in resources post-match, in school matching scenarios. We classify the effects of three types of errors inspired by real-world events on the students' admission outcomes. We provided estimates for the number of students harmed by different types of errors. We propose solutions to mitigate the detrimental impacts of errors on affected
parties. Our error mitigation strategies take into account desirable fairness properties and policy restrictions.

Our work can help inform policymakers on the design of compensatory processes that are fair to the wronged parties and assist legislators in providing a legal framework for compensating wronged parties.

\bibliographystyle{ACM-Reference-Format}
\bibliography{bibliography}

\appendix
\section{Full Proof of Expected Number of Envious Students Under Additive Error}
\label{sec:append}
Let A be the number of students who apply to the school during deferred acceptance, N be the number of \directHelp, and C be the capacity of the school. Let $B=A-C$ be the number of students eligible for a boost. We want to compute the expected number of \Envy when a school randomly boosts some unselected students. We know that exactly $N$ previously selected students will be jealous. This is because by boosting $N$ students we push out exactly $N$ of the top $C$ students directly. This is true no matter which students are boosted. For the remainder of the proof, we will omit the $N$ directly affected students and use $B$ instead of $A$.  We can easily account for these students at the end.  

Since we have countably many outcomes, the definition of expected value states this expectation will be:
$$
E(\#jealous) = \sum_{j=0}^{j=B-N}j \times P(\#jealous=j)
$$
The minimum number of jealous students is zero, if the the top N non-selected students are chosen. The maximum number of jealous students is everyone except the boosted students or $B-N$ if the lowest ranked student is boosted.

We can compute this number by realizing that there is a bijection between the number of jealous students and the rank of the lowest-ranked boosted student by the random admission function. If the lowest ranked admitted student has rank $N + j$, then there will be exactly $j$ students with justifiable envy. This is because every student ranked above the lowest ranked boosted student either admitted themselves or has justifiable envy of the worst ranked boosted student, and no students ranked below the worst ranked admitted student have justifiable envy. 

To compute the number of random selections that lead to exactly $j$ students with justifiable envy, we can use binomial coefficients. We do this by phrasing the problem as given that the the worst ranked student has rank $j + N$ how many ways can we select the $N-1$ other boosted students from the $j+N-1$ students ranked above the worst ranked admitted student. We can write this number as:
$$
\binom{j + N - 1}{N-1}
$$
Since in the error case all these selections are equally likely, the probability becomes this number divided by the total number of ways of selecting N students. We know we are not double-counting any groups since each term has a unique rank for the worst-ranked selected student. This gives:
$$
P(\#jealous=j) = \frac{\binom{j + N - 1}{N-1}}{\binom{B}{N}}
$$
Which means the expected value of $\#j$ is:
$$
 E(\#jealous) = \sum_{j=0}^{B-N} j \times \frac{\binom{j + N - 1}{N-1}}{\binom{B}{N}}
$$
This simplifies to:
$$
 E(\#jealous) = \sum_{j=0}^{B-N} j \times \frac{\binom{j + N - 1}{N-1}}{\binom{B}{N}} = (B-N) \times \frac{N}{N + 1}
$$
To show this we will need some binomial coefficient identities.
First, we have one based on the factorial definition of binomial coefficients:
\begin{equation} \label{lem1}
j \times \binom{j + N -1}{N-1} = N \times \binom{j + N -1}{j-1}
\end{equation}
This can be easily seen by writing out the factorial explicitly.
\begin{equation}\begin{split}
j \times \binom{j + N -1}{N-1} &= N \times \binom{j + N -1}{j-1}\\
j \times \frac{(j + N - 1)!}{(N-1)!(j)!} &= N \times \frac{(j + N - 1)!}{(j-1)!(N!)}\\
\frac{(j + N - 1)!}{(N-1)!(j-1)!} &= \frac{(j + N - 1)!}{(j-1)!(N-1)!}\\
\end{split}
\end{equation}
We will also use two traditional properties of binomial coefficients:
\begin{equation} \label{p1}
\binom{n+1}{k+1} = \binom{n}{k} + \binom{n}{k+1}
\end{equation}
and 
\begin{equation} \label{p2}
\binom{n}{k} = \binom{n}{n-k}
\end{equation}
We will need one more binomial identity
\begin{equation} \label{lem2}
\binom{n}{k+1} = \binom{n}{k}\frac{n-k}{k+1}
\end{equation}
This comes from the traditional binomial identity:
\begin{equation}
\binom{n}{k} = \binom{n}{k-1}\frac{n-k+1}{k}
 \end{equation}
We then substitute $k+1$ for $k$.

Lastly we will be using the hockey stick identity to collapse the sum which states (in one form):
\begin{equation} \label{p3}
\sum_{Q=0}^{B-N} \binom{Q + N}{Q} = \binom{B+1}{B-N}
\end{equation}
This can be proven with induction.

Now we are ready to begin the algebra, starting with our original claim.
$$
(B-N) \times \frac{N}{N + 1} = \sum_{j=0}^{B-N} j \times \frac{\binom{j + N - 1}{N-1}}{\binom{B}{N}}
$$
When $j=0$ the term clearly evaluates to 0 so we can remove that case.
$$(B-N) \times \frac{N}{N + 1} = \sum_{j=1}^{B-N} j \times \frac{\binom{j + N - 1}{N-1}}{\binom{B}{N}}$$
Now we apply \ref{lem1} to rewrite the binomial coefficient:
$$
(B-N) \times \frac{N}{N + 1} = \sum_{j=1}^{B-N} N \times \frac{\binom{j + N - 1}{j-1}}{\binom{B}{N}}
$$
Next pull the constants out of the summation.
$$
(B-N) \times \frac{N}{N + 1} = \frac{N}{\binom{B}{N}} \times \sum_{j=1}^{B-N} \binom{j + N - 1}{j-1}
$$
Now we will substitute $Q=j-1$, remembering to adjust the bounds of the summation:
$$
(B-N) \times \frac{N}{N + 1} = \frac{N}{\binom{B}{N}} \times \sum_{Q=0}^{B-N-1}\binom{Q + N}{Q}
$$
Next, we will rewrite the summation by both adding and subtracting $\binom{B}{B-N}$.
$$
(B-N) \times \frac{N}{N + 1} = \frac{N}{\binom{B}{N}} (\sum_{Q=0}^{B-N} \binom{Q + N}{Q} - \binom{B}{B-N})
$$
Now we apply \ref{p2}
$$
(B-N) \times \frac{N}{N + 1} = \frac{N}{\binom{B}{N}} (\sum_{Q=0}^{B-N} \binom{Q + N}{Q} - \binom{B}{N})
$$
Next, we simplify
$$
(B-N) \times \frac{N}{N + 1} = \frac{N}{\binom{B}{N}} (\sum_{Q=0}^{B-N} \binom{Q + N}{Q}) - N
$$
Now apply the hockey-stick identity from equation \ref{p3}
$$
(B-N) \times \frac{N}{N + 1} = \frac{N}{\binom{B}{N}} (\binom{B+1}{B-N}) - N
$$
Now apply \ref{p2} again
$$
(B-N) \times \frac{N}{N + 1} = \frac{N}{\binom{B}{N}} (\binom{B+1}{N+1}) - N
$$
Next, we apply the recursive definition of binomial coefficients \ref{p1}
$$
(B-N) \times \frac{N}{N + 1} = \frac{N}{\binom{B}{N}} (\binom{B}{N}  + \binom{B}{N+1}) - N
$$
Simplify again
$$
(B-N) \times \frac{N}{N + 1} = \frac{N}{\binom{B}{N}} (\binom{B}{N+1}) + N - N
$$
$$
(B-N) \times \frac{N}{N + 1} = \frac{N}{\binom{B}{N}} (\binom{B}{N+1})
$$
Now we apply equation \ref{lem2}
$$
(B-N) \times \frac{N}{N + 1} = \frac{N}{\binom{B}{N}} \times \binom{B}{N}\frac{B-N}{N+1}
$$
Simplify one last time and we have our answer
$$
(B-N) \times \frac{N}{N + 1} = (B-N) \times \frac{N}{N + 1}
$$

$$
(B-N) \times \frac{N}{N + 1} = \sum_{j=0}^{B-N} j \times \frac{\binom{j + N - 1}{N-1}}{\binom{B}{N}}
$$
Remember that we need to replace B with $A-C$ and add back the $N$ directly affected students to get the final answer of:
$$
N + (A-C-N) \times \frac{N}{N + 1}
$$

\end{document}